\documentclass[11pt,letterpaper,oneside,english]{article}
\usepackage[left=3.5cm,right=3.5cm,top=3cm,bottom=3cm]{geometry}

\usepackage{mdwlist,enumerate}

\makeatletter
\let\@font@warningori\@font@warning
\newcommand\shutup{\def\@font@warning##1{}}
\newcommand\youcanspeak{\let\@font@warning\@font@warningori}
\makeatother 
\usepackage{amsmath,amsbsy,bm,latexsym,amssymb}
\shutup
\usepackage{fourier}
\youcanspeak
\usepackage[scaled=0.875]{helvet}

\usepackage{graphics, subfigure, float}
\usepackage{fp, calc}

\usepackage[latin1]{inputenc}
\usepackage[american]{babel}
\usepackage[T1]{fontenc} 

\usepackage{amscd,amsthm}

\usepackage{verbatim, comment}

\usepackage[dvips]{graphicx,epsfig,color}
\usepackage{pst-all}
\usepackage{pstricks-add}

\theoremstyle{plain}
\newtheorem{theorem}{Theorem}
\newtheorem*{theorem*}{Theorem}
\newtheorem{corollary}[theorem]{Corollary}

\newtheorem{lemma}[theorem]{Lemma}
\newtheorem*{lemma*}{Lemma}

\newtheorem*{claim*}{Claim}

\newtheorem*{conjecture*}{Conjecture}
\newtheorem*{problem*}{Problem}
\newtheorem{definition}{Definition}
\newtheorem*{definition*}{Definition}

\theoremstyle{remark}
\newtheorem{remark}{Remark}
\newtheorem*{remark*}{Remark}

\newtheorem*{algorithm*}{Algorithm}

\providecommand{\setZ}{\mathbb{Z}}

\providecommand{\setR}{\mathbb{R}}

\newcommand{\conv}{\textrm{conv}}

\newcommand{\proj}{\textrm{proj}}
\newcommand{\supp}{\textrm{supp}}
\newcommand{\bsupp}{\overline{\textrm{supp}}}
\def\tchoose#1#2{ {\textstyle{{#1} \choose {#2}}} }

\floatstyle{ruled}
\newfloat{algorithm}{tbp}{loa}
\floatname{algorithm}{Algorithm}
\newfloat{algorithm*}{tbp}{loa}
\floatname{algorithm*}{Algorithm}

\usepackage[displaymath,textmath,sections,graphics, subfigure, floats]{preview} 
\PreviewEnvironment{center} 
\PreviewEnvironment{pspicture} 
        \def\drawRect#1#2#3#4#5{
           \FPeval{\x2}{#2 + #4} 
           \FPeval{\y2}{#3 + #5} 
           \pspolygon[#1](#2,#3)(\x2,#3)(\x2,\y2)(#2,\y2)
        }

\psset{arrowsize=6pt, labelsep=2pt, linewidth=1.5pt}
\makeatother

\begin{document}

\title{Directed Steiner Tree and the Lasserre Hierarchy}

\author{Thomas Rothvoß\thanks{Supported by the Alexander von Humboldt Foundation within the Feodor Lynen program, by ONR grant N00014-11-1-0053 and by NSF contract
CCF-0829878.}\vspace{3mm}\\M.I.T. \\
{ \tt{rothvoss@math.mit.edu}} \vspace{3mm} }

\maketitle

\begin{abstract}
\noindent The goal for the \textsc{Directed Steiner Tree} problem is to 
find a minimum cost tree in a directed graph $G=(V,E)$ that connects all 
terminals $X$ to a given root $r$. It is well known that modulo a logarithmic factor
it suffices to consider acyclic graphs where the nodes are arranged in $\ell \leq \log |X|$
levels. Unfortunately the natural LP formulation has a $\Omega(\sqrt{|X|})$ integrality gap
already for $5$ levels. 
We show that for every $\ell$, the \emph{$O(\ell)$-round Lasserre Strengthening} of this LP 
has integrality gap $O(\ell \log |X| )$. This provides a polynomial time $|X|^{\varepsilon}$-approximation
and a $O(\log^3 |X| )$ approximation in 
$O(n^{\log |X|})$ time, matching the best known approximation guarantee obtained by a greedy algorithm of Charikar et al.
\end{abstract}

\section{Introduction}

Most optimization problems that appear in  combinatorial optimization 
can be written as an integer linear problem,  
say in the form $\min\{ c^Tx \mid Ax \geq b; \; x \in \{ 0,1\}^n \}$, where the system $Ax\geq b$
represents the problem structure. Since $c$ is linear, this is 
optimizing over the convex hull $K_I := \conv(K \cap \{ 0,1\}^n)$, where $K := \{ x \in \setR^n \mid Ax \geq b \}$ denotes a polyhedron.

Such optimization problems are $\mathbf{NP}$-hard in general, thus a standard approach 
for obtaining approximate solutions
is to optimize instead over the relaxation $K$ and then try to extract 
a close integral solution. 
This approach yields in many cases solutions whose quality matches the lower bound 
provided by the PCP Theorem or the Unique Games Conjecture (which is the case e.g. for {\sc Set Cover}~\cite{SetCoverIntGap-Lovasz75,SetCover-lnN-hardness-FeigeJACM98}, {\sc Vertex Cover}~\cite{VertexCover-2-UGC-hardness-KhotRegevJCSS08} and {\sc Facility Location}~\cite{FacilityLocation1.463-hardness-GuhaKhuller-JALG99,FacilityLocation-1.488-apx-Li-ICALP2011}. 
However, there is  a significant number of 
problems, where the \emph{integrality gap} between $K$ and $K_I$ appears to be far higher
than the approximability of the problem, so that a stronger formulation is needed. 

At least in the field of approximation algorithm, researchers have so far mostly preferred 
problem-specific inequalities to lower the integrality gap
(a nice example is the $O(1)$-apx for {\sc Min-Sum Set Cover}~\cite{MinSumSetCover-BansalGuptaKrishnaswamySODA2010}).
However there are very general techniques that
can be used to strengthen the convex relaxation $K$.

Especially in the field of (computational) integer programming, the approach of
cutting planes is very popular. In the \emph{Gomory-Chvátal Closure} $CG(K) \subseteq K$ 
one adds simultaneously cuts  $a^Tx \leq \lfloor\beta\rfloor$ for all valid inequalities $a^Tx \leq \beta$ with $a \in \setZ^n$. 
On the positive side, after at most $O(n^2 \log n)$ 
iterative applications of the closure operation, one reaches $K_I$~\cite{CG-rank-is-n2logN-EisenbrandSchulz-IPCO99} (assuming that $K\subseteq[0,1]^n$). But the drawback is 
that already optimizing over the first closure $CG(K)$ is $\mathbf{coNP}$-hard~\cite{OptimizingOverCGisCoNPhard-Eisenbrand-Combinatorica99}. Singh and Talwar~\cite{CG-ranksForkdimMatchingAndVC-SinghTalwar-APPROX2010} studied the effect of Gomory-Chvátal cuts to the integrality gap of hypergraph matchings and other problems .

However, more promising for the sake of approximation algorithms are probably LP/SDP hierarchies like the ones of
\emph{Balas, Ceria, Cornuéjols}~\cite{BalasCeriaCornuejols-Hierarchy-MathProg93}; 
\emph{Lovász, Schrijver}~\cite{LovaszSchrijverHierarchy91} (with LP-strengthening $LS$ and an SDP-strengthening $LS_+$); \emph{Sherali, Adams}~\cite{SheraliAdamsHierarchy1990}
or \emph{Lasserre}~\cite{ExplicitExactSDP-Lasserre-IPCO01,GlobalOpt-Lasserre01}.
  On the $t$-th level, they all use $n^{O(t)}$ additional variables to strengthen $K$
(thus the term \emph{Lift-and-Project Methods}) and 
they all can be solved in time $n^{O(t)}$. Moreover, for $t=n$ they define the integral hull $K_I$ and
for any set of $|S| \leq t$ variables, a solution $x$ can be written as convex combinations of vectors 
from $K$ that are integral on $S$. Despite these similarities, the Lasserre SDP relaxation is strictly stronger
than all the others. We refer to the survey of Laurent~\cite{SDP-hierarchies-Survey-Laurent2003}
for a detailed  comparison.

Up to now, there have been few (positive) result on the use of hierarchies 
in approximation algorithms. 
One successful application of Chlamtá{\v{c}}~\cite{n0.20-apx-for-3colgraphsChlamtacFOCS07} uses the 3rd level 
of the Lasserre relaxation to find $O(n^{0.2072})$-colorings
for 3-colorable graphs. 
It lies in the range of possibilities that $O(\log n)$ levels of Lasserre might be enough to obtain 
a coloring with $O(\log n)$ colors in 3-colorable graphs~\cite{ChromaticNumberApx-AroraCharikarChlamtacSTOC06}. 
In fact, for special graph classes, there has been recent progress by Arora and Ge \cite{GraphColoringForSpecialClassesViaLasserreAroraGeAPPROX11}.
Chlamtá{\v{c}} and Singh~\cite{IndSet-in-3-regular-graphs-ChlamtacSingh-APPROX08} showed that $O(1/\gamma^2)$ rounds of a mixed hierarchy 
can be used to obtain an independent set of size $n^{\Omega(1/\gamma^2)}$ in a 3-uniform hypergraph,
whenever it has an independent set of size $\gamma n$.
After a constant number of rounds of Sherali-Adams, the integrality gap
for the matching polytope
reduces to $1+\varepsilon$~\cite{SA-Relaxation-for-matching-MathieuSinclair-STOC09}.
The same is true for {\sc MaxCut} in dense graphs (i.e. graphs with $\Omega(n^2)$ edges) \cite{LPs-for-MaxCut-Vega-Kenyon-Mathieu-SODA07}.
The Sherali-Adams hierarchy is also used in~\cite{MaxMinAllocationViaDegreeLowerBoundedArb-BateniCharikarGuruswami-STOC09} to find degree lower-bounded arborescences. 

Guruswami and Sinop provide approximation algorithms for quadratic integer programming problems
whose performance guarantees depend on the eigenvalues of the graph Laplacian \cite{Lasserre-and-QIP-GuruswamiSinopECCC11}. Also the Lasserre-based approach of
\cite{SDP-rounding-via-global-correlation-BarakRaghavendraSteurerFOCS2011} 
for {\sc Unique Games} depends on the eigenvalues of the underling graph adjacency matrix. 
Though the $O(\sqrt{\log n})$-apx of Arora, Rao and Vazirani~\cite{SparsestCut-sqrtLogN-apx-AroraRaoVazirani-STOC2004} for {\sc Sparsest Cut} does not explicitly use 
hierarchies, their triangle inequality is implied by $O(1)$ rounds of Lasserre. 
For a more detailed overview on the use of hierarchies in approximation
algorithms, see the recent survey of Chlamtá{\v{c}} and Tulsiani~\cite{ConvexRelaxations-survey-Chlamtac-Tulsiani}.

Moreover, integrality gap lower bounds exist for various problems~\cite{Lasserre-and-Cut-polytope-Laurent-MOR03,LS-SDP-lower-bounds-for-Max3Sat-kVC-SetCover-STOC2005,LS-VC-lower-bound-Tourlakis2006,LinearRound-7over6-gap-for-LS-SDP-CCC07,2gap-for-LS-LP-after-OmegaN-rounds-STOC2007,LS-SDP-gap-2-for-VC-FOCS2007,LinearLowerBoundLasserre-Schoenebeck-FOCS2008,SA-Gaps-for-MaxCut-VC-SparsestCut-STOC2009,LinearRoundGaps-Lasserre-Gaps-Tulsiani-STOC09}. To name only few of these results, even a linear
number of Lasserre rounds cannot refute unsatisfiable constraint satisfaction problems~\cite{LinearLowerBoundLasserre-Schoenebeck-FOCS2008} and the gap for {\sc Graph Coloring} is still 
$k$ versus $2^{\Omega(k)}$~\cite{LinearRoundGaps-Lasserre-Gaps-Tulsiani-STOC09}. 
In contrast, the LP-based hierarchies LS and SA cannot even reduce 
the {\sc MaxCut} gap below $2-\varepsilon$ after $\Omega(n)$~\cite{2gap-for-LS-LP-after-OmegaN-rounds-STOC2007} and $n^{\delta}$~\cite{SA-Gaps-for-MaxCut-VC-SparsestCut-STOC2009} many rounds, respectively.
Recall that already a single round of the SDP based hierarchies reduces the gap to $1.13$.

In this paper, we apply the Lasserre relaxation to the flow-based linear programming relaxation 
of {\sc{Directed Steiner Tree}}.
The input for this problem  consists of a  directed
graph $G=(V,E)$ with edge cost $c : E \to \setR_+$, a \emph{root} $r \in V$ and \emph{terminals} $X \subseteq V$. 
The goal is to compute a subset $T \subseteq E$ of edges such that there is an $r$-$s$
path in $T$ for each terminal $s$. Note that the cheapest such set always forms a tree
(see Figure~\ref{fig:DST-example}). 

By a straightforward reduction from {\sc Set Cover} one easily sees that the problem 
is $\Omega(\log n)$-hard~\cite{SetCover-lnN-hardness-FeigeJACM98} to approximate.
Zelikovsky~\cite{Zelikovsky97aseries} obtained a $|X|^{\varepsilon}$-approximation for every constant $\varepsilon > 0$ using a greedy approach.
In fact, he provided the following useful insight:
\begin{theorem}[\cite{Zelikovsky97aseries,PolymatroidSteinerTreeCalinecuZelikovsky2005}] \label{thm:lLevelTree}
For every $\ell \geq 1$, there is a tree $T$ (potentially using edges in the metric closure)
of cost $c(T) \leq \ell \cdot |X|^{1/\ell}\cdot OPT$ such that every $r$-$s$ path (with $s \in X$) in $T$ contains
at most $\ell$ edges. 
\end{theorem}
In other words, at the cost of a factor\footnote{
The claim in \cite{Zelikovsky97aseries} was initially $|X|^{1/\ell}$, which
is incorrect, though it was later on heavily used in the literature. 
In a later paper, 
Calinescu and Zelikovsky\cite{PolymatroidSteinerTreeCalinecuZelikovsky2005} change the claim to $\ell\cdot |X|^{1/\ell}$. 
As a consequence, the $O(\log^2 |X|)$-apx in \cite{DirectedSteinerTree-polylog-qpolytime-apx-CharikarEtAl-Journal99} has to be changed to a $O(\log^3 |X|)$-apx.} $\ell \cdot |X|^{1/\ell}$ in the approximation guarantee, one 
may assume that the graph is acyclic and the nodes are arranged in $\ell$ \emph{levels}\footnote{This can easily be achieved by taking $\ell+1$ copies of the node set in the
original graph and insert cost-$0$ edges between copies of the same node.
} (observe that for $\ell = \log |X|$, one has $\ell \cdot |X|^{1/\ell} = \ell \cdot (2^{\log |X|})^{1/\log |X|}= O(\log |X|)$).
\begin{figure}
\begin{center}
\psset{xunit=1cm,yunit=1.0cm,linewidth=1pt}
\begin{pspicture}(0,0.8)(5,4.0)
  \drawRect{linearc=0.1,fillcolor=lightgray,fillstyle=solid,linewidth=0.5pt,linecolor=gray}{0.5}{0.7}{4}{0.6}
  \drawRect{linearc=0.1,fillcolor=lightgray,fillstyle=solid,linewidth=0.5pt,linecolor=gray}{0.5}{1.7}{4}{0.6}
  \drawRect{linearc=0.1,fillcolor=lightgray,fillstyle=solid,linewidth=0.5pt,linecolor=gray}{0.5}{2.7}{4}{0.6}
  \drawRect{linearc=0.1,fillcolor=lightgray,fillstyle=solid,linewidth=0.5pt,linecolor=gray}{0.5}{3.7}{4}{0.6}
  \cnodeput[framesep=1.5pt,fillstyle=solid,fillcolor=white](2.5,4){r}{$r$}
  \cnodeput[framesep=5pt,fillstyle=solid,fillcolor=white](1.5,3){u1}{} 
  \cnodeput[framesep=5pt,fillstyle=solid,fillcolor=white](2.5,3){u2}{}
  \cnodeput[framesep=5pt,fillstyle=solid,fillcolor=white](3.5,3){u3}{}
  \cnodeput[framesep=5pt,fillstyle=solid,fillcolor=white](1,2){v1}{}
  \cnodeput[framesep=5pt,fillstyle=solid,fillcolor=white](2,2){v2}{}
  \cnodeput[framesep=5pt,fillstyle=solid,fillcolor=white](3,2){v3}{}
  \cnodeput[framesep=5pt,fillstyle=solid,fillcolor=white](4,2){v4}{}
 
  \fnode[linewidth=1.0pt,framesize=12pt,fillstyle=solid,fillcolor=white](1,1){s1} 
  \fnode[linewidth=1.0pt,framesize=12pt,fillstyle=solid,fillcolor=white](2,1){s2}
  \fnode[linewidth=1.0pt,framesize=12pt,fillstyle=solid,fillcolor=white](3,1){s3}
  \fnode[linewidth=1.0pt,framesize=12pt,fillstyle=solid,fillcolor=white](4,1){s4}
  \ncline[arrowsize=5pt,linecolor=gray]{->}{r}{u2} \nbput[labelsep=0pt,npos=0.5]{$\gray{4}$}
  \ncline[arrowsize=5pt,linecolor=gray]{->}{u2}{v2} \naput[labelsep=0pt,npos=0.5]{$\gray{8}$}
  \ncline[arrowsize=5pt,linecolor=gray]{->}{u2}{v3} \naput[labelsep=0pt,npos=0.7]{$\gray{9}$}
  \ncline[arrowsize=5pt,linecolor=gray]{->}{u2}{v4} \naput[labelsep=0pt,npos=0.3]{$\gray{7}$}
  \ncline[arrowsize=5pt,linecolor=gray]{->}{v2}{s1} \nbput[labelsep=0pt,npos=0.5]{$\gray{6}$}
  \ncline[arrowsize=5pt,linecolor=gray]{->}{v3}{s2} \naput[labelsep=0pt,npos=0.8]{$\gray{7}$}
  \ncline[arrowsize=5pt,linecolor=gray]{->}{v3}{s3} \naput[labelsep=0pt,npos=0.5]{$\gray{4}$}
  \ncline[arrowsize=5pt,linecolor=gray]{->}{v3}{s4} \naput[labelsep=0pt,npos=0.5]{$\gray{8}$}

  \ncline[arrowsize=5pt,linewidth=1.5pt]{->}{r}{u1} \nbput[labelsep=0pt,npos=0.5]{$3$}
  \ncline[arrowsize=5pt,linewidth=1.5pt]{->}{r}{u3} \naput[labelsep=0pt,npos=0.5]{$2$}
  \ncline[arrowsize=5pt,linewidth=1.5pt]{->}{u1}{v1} \nbput[labelsep=0pt,npos=0.5]{$3$}
  \ncline[arrowsize=5pt,linewidth=1.5pt]{->}{u1}{v2} \naput[labelsep=0pt,npos=0.5]{$5$}
  \ncline[arrowsize=5pt,linewidth=1.5pt]{->}{u3}{v4} \naput[labelsep=0pt,npos=0.5]{$2$}
  \ncline[arrowsize=5pt,linewidth=1.5pt]{->}{v1}{s1} \nbput[labelsep=0pt,npos=0.5]{$2$}
  \ncline[arrowsize=5pt,linewidth=1.5pt]{->}{v2}{s2} \nbput[labelsep=0pt,npos=0.5]{$0$}
  \ncline[arrowsize=5pt,linewidth=1.5pt]{->}{v2}{s3} \naput[labelsep=0pt,npos=0.3]{$1$}
  \ncline[arrowsize=5pt,linewidth=1.5pt]{->}{v4}{s4} \naput[labelsep=0pt,npos=0.5]{$1$}
  \rput[c](-1,1){{$V_{\ell}:$}}
  \rput[c](-1,2){{$V_{\ell-1}:$}}
  \rput[c](-1,2.5){\Large{$\vdots$}}
  \rput[c](-1,3){{$V_{1}:$}}
  \rput[c](-1,4){{$V_{0}:$}}
\end{pspicture}
\caption{Illustration of layered graph $G$ with $\ell = 3$. Edges are labelled with their cost. Black edges denote an optimum {\sc Directed Steiner Tree} solution. Rectangles depict terminals.\label{fig:DST-example}}
\end{center}
\end{figure}
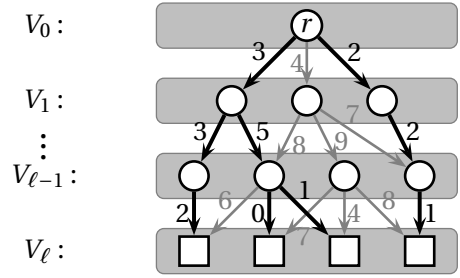
Later Charikar, Chekuri, Cheung, Goel, Guha and Li~\cite{DirectedSteinerTree-polylog-qpolytime-apx-CharikarEtAl-Journal99}) gave a $O(\log^3 |X|)$ approximation in time $n^{O(\log |X|)}$ again using a sophisticated greedy algorithm.
So far, methods based on linear programming were less successful. In fact, Zosin and Khuller~\cite{DirectedSteinerTree-IntegralityGap-ZosinKhuller-SODA2002}
show that the natural flow based LP relaxation has an integrality gap of $\Omega(\sqrt{|X|})$. 
The importance of {\sc Directed Steiner Tree} lies in the fact that it generalizes a huge
number of problems, e.g. {\sc Set Cover}, (non-metric, multilevel) {\sc Facility Location} and   $\textsc{Group Steiner Tree}$.
For the latter problem, the input consists of an \emph{undirected} (weighted) graph $G=(V,E)$, 
groups $G_1,\ldots,G_k \subseteq V$ of terminals and a root $r \in V$. The goal here is to find a tree $T \subseteq E$ of minimum cost
$c(T)$ that connects at least one terminal from every group to the root.
The state of the art for $\textsc{Group Steiner Tree}$ is still the elegant approach 
of Garg, Konjevod and Ravi~\cite{GroupSteinerTreePolylogApx-GargKonjevodRavi-2000}. Using 
 $O(\log n)$-average distortion tree embeddings, one can assume that the input graph itself is a tree. 
Then \cite{GroupSteinerTreePolylogApx-GargKonjevodRavi-2000} solve a flow based LP and provide a 
 rounding scheme, which gives a $O(\log^2 n)$-approximation w.r.t. the optimum solution in the 
tree graph. 
Surprisingly, \cite{IntegralityGapForGroupSteinerTree-and-DirectedSteinerTree-HalperinEtAl-SODA03} found a tree instance, which indeed has an integrality gap of
 $\Omega(\log^2 n)$. Later Halperin and Krauthgamer~\cite{PolylogarithmicInapproximabilityForGroupSteinerTree-HalperinKrauthgamer-STOC03} even 
proved a $\Omega(\log^{2-\varepsilon} n)$ inapproximability for {\sc Group Steiner Tree} (on tree graphs) 
and in turn also for  {\sc Directed Steiner Tree}.

\subsection*{Our contribution}

Many researchers have failed in designing stronger LP relaxations for
{\sc Directed Steiner Tree} (see Alon, Moitra and Sudakov~\cite{GraphsWithLargeInducedMatchings-AlonMoitraSudakov-STOC12} for a counterexample to a promising
approach). We make partial progress by
showing that in an $\ell$-level graph, already $O(\ell)$ rounds of the Lasserre hierarchy
drastically reduce the integrality gap of the natural flow-based LP for
{\sc Directed Steiner Tree} from $\Omega(\sqrt{|X|})$ (for $\ell\geq5$) down to $O(\ell \log |X| )$.
This gives an alternative polylogarithmic approximation in quasi-polynomial time (and the
first one that is based on convex relaxations). 

In this paper, we try to promote the application of hierarchies in approximation algorithms. 
For this sake, we demonstrate how the Lasserre relaxation can be used as a black box
in order to obtain  powerful (yet reasonably simple) approximation algorithms.


From a technical view point we adapt the 
rounding scheme of Garg, Konjevod and Ravi~\cite{GroupSteinerTreePolylogApx-GargKonjevodRavi-2000}.
Though their algorithm and analysis crucially relies
on the fact that the input graph itself is a tree, it turns out that one can use 
instead the values of the auxiliary variables
in the $O(\ell)$-round SDP to perform the rounding. Another ingredient
for our analysis is the recent \emph{Decomposition Theorem} of 
Karlin, Mathieu and Nguyen~\cite{IntGapLasserre-for-Knapsack-IPCO2011}
for the Lasserre hierarchy.

\section{The Lasserre Hierarchy}

In this section, we provide a definition of the Lasserre hierarchy and all 
properties that are necessary for our purpose. In our notation, we mainly follow the 
  survey of Laurent~\cite{SDP-hierarchies-Survey-Laurent2003}. 
Let $\mathcal{P}_t([n]) := \{ I \subseteq [n] \mid |I| \leq t\}$ be the set of all index sets of
cardinality at most $t$ and let $y \in \setR^{\mathcal{P}_{2t+2}([n])}$ be a vector with entries 
$y_I$ for all $I\subseteq[n]$ with $|I| \leq 2t+2$. Intuitively $y_{\{i\}}$ represents the original
variable $x_i$ and the new variables $y_I$ represent $\prod_{i\in I} x_i$. 
We define the \emph{moment matrix} $M_{t+1}(y) \in \setR^{\mathcal{P}_{t+1}([n]) × \mathcal{P}_{t+1}([n])}$ by
\[
  (M_{t+1}(y))_{I,J} := y_{I\cup J} \quad \forall |I|,|J| \leq t+1.
\]
For a linear constraint $a^Tx \geq \beta$ with $a \in \setR^n$ and $\beta \in \setR$ we define ${a \choose \beta} * y$ as the
vector $z$ 
with $z_I := \sum_{i\in[n]} a_iy_{I\cup\{i\}}-\beta y_{I}$.
\footnote{This notation was initially introduced 
for multivariate degree-one polynomials $g(x) = \sum_{I} g_I \prod_{i\in I} x_i$ which induce constraints
of the form $g(x) \geq 0$. 
In this general case, one defines $(g * y)_I := \sum_{K \subseteq [n]} g_K \cdot y_{I \cup K}$. Note that 
for a linear constraint $a^Tx \geq \beta$, one has $g_{\{i\}}=a_i, g_{\emptyset}=-\beta$ and $g_I=0$ for $|I|>1$. 
We stick to this notation to be consistent with the existing literature.} 

\begin{definition}
Let $K = \{ x \in \setR^n \mid Ax \geq b\}$. We define the \emph{$t$-th level of the Lasserre hierarchy} 
$\textsc{Las}_t(K)$ as the set of vectors $y \in \setR^{\mathcal{P}_{2t+2}([n])}$ that
satisfy 
\[
  M_{t+1}(y) \succeq 0; \hspace{1cm} M_{t}(\tchoose{A_{\ell}}{b_{\ell}} * y) \succeq 0 \quad \forall\ell\in[m]; \hspace{1cm} y_{\emptyset}=1.
\]
Furthermore, let  $\textsc{Las}_t^{\textrm{proj}} := \{ (y_{\{1\}},\ldots,y_{\{n\}}) \mid y \in \textsc{Las}_t(K) \}$ be the projection on the original variables. 
\end{definition}
Intuitively, the PSD-constraint $M_{t}(\tchoose{A_{\ell}}{b_{\ell}} * y) \succeq 0$ guarantees
that $y$ satisfies the $\ell$-th linear constraint, while $M_{t+1}(y) \succeq 0$ takes care
that the variables are consistent (e.g. it guarantees that $y_{\{1,2\}} \in [y_{\{1\}}+y_{\{2\}}-1, \min\{ y_{\{1\}},y_{\{2\}}\}]$).
The Lasserre hierarchy can even be applied to non-convex semi-algebraic sets -- but for the sake
of a simple presentation we stick to polytopes.
Fortunately, one can use the Lasserre relaxation conveniently as a black-box. We list all properties, that 
we need for our approximation algorithm:

\begin{theorem} \label{thm:PropertiesOfLasT}
Let $K = \{ x \in \setR^n \mid Ax \geq b\}$ and $y \in {\textsc{Las}}_t(K)$. Then the following holds: 
\begin{enumerate}[(a)] \addtolength{\itemsep}{-0.4\baselineskip}
\item \label{item:PropertiesOfLasT-Hierarchy} $\conv(K \cap \{ 0,1\}^n) = \textsc{Las}_n^{\textrm{proj}}(K) \subseteq \textsc{Las}_{n-1}^{\textrm{proj}}(K) \subseteq \ldots \subseteq \textsc{Las}_0^{\textrm{proj}}(K) \subseteq K$.
\item \label{item:PropertiesOfLasT-Monotone} One has $0 \leq y_I \leq y_J \leq 1$ for all $I\supseteq J$ with  $0\leq|J|\leq |I| \leq t$.
\item \label{item:PropertiesOfLasT-Cond-for-yI0} Let $I \subseteq [n]$ with $|I| \leq t$. Then $K \cap \{ x \in \setR^n \mid x_i =1 \; \forall i \in I\} = \emptyset \; \Longrightarrow \; y_I = 0$.
\item \label{item:PropertiesOfLasT-IntOnI} Let $I \subseteq [n]$ with $|I| \leq t$. Then $y \in \conv(\{z \in \textsc{Las}_{t-|I|}(K) \mid z_{\{i\}} \in \{ 0,1\} \; \forall i\in I\})$.\footnote{Formally spoken, vectors in $\textsc{Las}_{t- |I|}$ have less dimensions than $y$. Thus it would be more 
correct to write $z_{|\mathcal{P}_{2(t - |I| ) + 2}([n])} \in \textsc{Las}_{t - |I|}(K)$ where
$z_{|\mathcal{P}_{2t+2-2|I|}([n])}$ denotes the restriction of $z$ to all entries $I$ with $|I| \leq 2(t - |I| ) + 2$.}
\item \label{item:PropertiesOfLasT-DecThm} Let $S \subseteq [n]$ be a subset of variables 
such that $\max\{ |I| : I \subseteq S; x \in K; x_i = 1\; \forall i \in I\} \leq k < t$.
Then $y \in \conv(\{z \in \textsc{Las}_{t-k}(K) \mid z_{\{i\}} \in \{ 0,1\} \; \forall i\in S\})$.
\item \label{item:PropertiesOfLasT-yI1-iff-yi1-for-all-i-in-I} For any $|I| \leq t$ one has $y_I=1 \Leftrightarrow \bigwedge_{i\in I} (y_{\{i\}}=1)$.
\item \label{item:PropertiesOfLasT-product} For $|I| \leq t$: $(\forall i\in I : y_{\{i\}} \in \{ 0,1\}) \; \Longrightarrow \; y_I = \prod_{i\in I} y_{\{i\}}$.
\item \label{item:PropertiesOfLasT-yI1-implies-yIJ-is-yJ} Let $|I|,|J| \leq t$ and $y_{I}=1$. Then $y_{I \cup J} = y_J$.
\end{enumerate}
\end{theorem}
\begin{proof}
Proofs of \eqref{item:PropertiesOfLasT-Hierarchy},\eqref{item:PropertiesOfLasT-Monotone},\eqref{item:PropertiesOfLasT-Cond-for-yI0} can be found in Laurent~\cite{SDP-hierarchies-Survey-Laurent2003}. 
\eqref{item:PropertiesOfLasT-DecThm} is the Decomposition Theorem of \cite{IntGapLasserre-for-Knapsack-IPCO2011}.
\eqref{item:PropertiesOfLasT-yI1-iff-yi1-for-all-i-in-I} and \eqref{item:PropertiesOfLasT-product} follow easily from \eqref{item:PropertiesOfLasT-Monotone} and \eqref{item:PropertiesOfLasT-yI1-implies-yIJ-is-yJ}.
For \eqref{item:PropertiesOfLasT-yI1-implies-yIJ-is-yJ},
consider the principal submatrix 
\[
 M = \begin{pmatrix} 
 1 & 1 & y_J \\
 1 & 1 & y_{I\cup J} \\
 y_{J} & y_{I\cup J} & y_J
\end{pmatrix}
\]
of $M_{t+1}(y)$ that is induced by indices $\{ \emptyset,I,J\}$ (substituting $y_I$ with $1$).
Then $\det(M) = -(y_J-y_{I\cup J})^2 \geq 0$ implies that $y_J = y_{I \cup J}$.
\end{proof}
Though all these properties are well known, to be fully self contained, we provide a complete
introduction with proofs of the non-trivial statements \eqref{item:PropertiesOfLasT-Hierarchy},\eqref{item:PropertiesOfLasT-IntOnI},\eqref{item:PropertiesOfLasT-DecThm} in the appendix. 

Especially  \eqref{item:PropertiesOfLasT-DecThm} is a remarkably strong property that does not
hold for the Sherali-Adams or Lovász-Schrijver hierarchy~(see \cite{IntGapLasserre-for-Knapsack-IPCO2011}).
For example, it implies that after $t=O(\frac{1}{\varepsilon})$ rounds, the 
integrality gap for the {\sc Knapsack} polytope is bounded by $1+\varepsilon$ (taking $S$ as all items that have profit
at least $\varepsilon\cdot OPT$). The same bound holds for the {\sc Matching} polytope $\{ x \in \setR_+^E \mid x(\delta(v)) \leq 1 \; \forall v \in V\}$
(since Property  \eqref{item:PropertiesOfLasT-DecThm} implies all Blossom inequalities up to $2t+1$ nodes). 
Another immediate consequence is that the {\sc Independent Set} polytope $\{ x\in \setR_+^V \mid x_{u}+x_v\leq1 \; \forall \{u,v\} \in E\}$ 
describes the integral hull after $\alpha(G)$ rounds of Lasserre (where $\alpha(G)$ is the stable set number of the considered graph).

\section{The linear program}

The natural LP formulation for {\sc Directed Steiner Tree} sends a unit flow from the root to each terminal 
$s \in X$ (represented by variables $f_{s,e}$). The amount of capacity that has to be paid on edge $e$ is 
  $y_e = \max\{ f_{s,e} \mid s \in X \}$. We abbreviate $\delta^+(v) := \{ (v,u) \mid (v,u) \in E\}$ ($\delta^-(v) := \{ (u,v) \mid (u,v) \in E\}$, resp.) 
as the edges outgoing (ingoing, resp.) from $v$ and $y(E') := \sum_{e \in E'} y_e$. The LP is
\begin{eqnarray*}
  \min \sum_{e\in E} c_ey_e & & \\
  \sum_{e \in \delta^+(v)} f_{s,e} - \sum_{e \in \delta^-(v)} f_{s,e} &=& \begin{cases} 1 & v=r \\
-1 & v=s \\
0 & \textrm{otherwise} \end{cases} \quad \forall s\in X\; \forall v \in V\\
  f_{s,e} &\leq& y_e \quad \forall s \in X \; \forall e\in E  \\
  y(\delta^-(v)) &\leq& 1 \quad \forall v \in V \\
  0\leq y_e &\leq& 1 \quad \forall e \in E \\
  0 \leq f_{s,e} &\leq& 1 \quad \forall s \in X \; \forall e \in E
\end{eqnarray*}
Note that we have an additional constraint  $y(\delta^-(v)) \leq 1$ (i.e. only one ingoing edge for each node) 
which is going to help us in the analysis. 
Let $K \subseteq \setR^E × \setR^{X × E}$ be the set of fractional solutions. 
This LP has an integrality gap\footnote{Unfortunately, the instance has a number of nodes which is exponential in the number of terminals. Of course, the instance of \cite{IntegralityGapForGroupSteinerTree-and-DirectedSteinerTree-HalperinEtAl-SODA03} provides a $\Omega(\log^2 n)$ gap. To the best of our
knowledge, there is no known instance with a $\omega(\log^2 n)$ integrality gap.} of $\Omega(\sqrt{|X|})$ even if the number of layers is  $5$~\cite{DirectedSteinerTree-IntegralityGap-ZosinKhuller-SODA2002}. 
From now on, we make the choice $t := 2\ell$, i.e. we consider the 
$2\ell$-round Lasserre strengthening of the above LP. 
Let $\mathcal{V}_t = \{ (P,H) \mid P \subseteq E; H \subseteq X × E; |P| + |H| \leq 2t+2 \}$ be the set of variable indices for the $t$-th level of Lasserre.
In other words  $\textsc{Las}_t(K) \subseteq [0,1]^{\mathcal{V}_t}$.
Let $Y = (Y_{P,H})_{(P,H) \in \mathcal{V}_t} \in \textsc{Las}_t(K)$ be an optimum solution for the Lasserre relaxation, which can be computed in time $n^{O(t)}$.
We abbreviate $OPT_f := \sum_{e\in E} c_ey_{\{e\}}$ as the objective function value. 

We will only address either groups of $y_e$ variables (then we write $y_H := Y_{H,\emptyset}$ for $H \subseteq E$),
or we address groups of $f_{s,e}$ variables for the same terminal $s \in X$. Then we write
 $f_{s,H} := Y_{\emptyset,\{ (s,e) \mid e \in H\}}$.

\section{The rounding algorithm}

By Theorem~\ref{thm:lLevelTree}, we may assume that the node set is partitioned 
into \emph{levels} $V_0=\{r\},V_1,\ldots,V_{\ell-1},V_{\ell}=X$ and all edges are running between consecutive layers (i.e. $E \subseteq \bigcup_{j=1}^{\ell} (V_{j-1} × V_{j})$). See Figure~\ref{fig:DST-example} for an illustration.
In the following, we will present an adaptation of the \cite{GroupSteinerTreePolylogApx-GargKonjevodRavi-2000}
rounding scheme to sample a set $T$ of paths from a distribution that depends on $Y$. 
For this sake, starting at layer $0$, we will 
go through all layers and for each path $P$ (ending in node $u$) that is sampled so far, 
we will extend it to $P \cup \{ (u,v) \}$ with probability $\frac{y_{P\cup \{ (u,v)\}}}{y_P}$.
\begin{enumerate*}
\item[(1)] $T := \emptyset$
\item[(2)] FOR ALL $e \in \delta^+(r)$ DO 
  \begin{enumerate*}
  \item[(3)] independently, with prob.  $y_{\{e\}}$, add path $\{e\}$ to $T$ 
  \end{enumerate*}
\item[(4)] FOR $j=1,\ldots,\ell-1$ DO 
  \begin{enumerate*}
  \item[(5)] FOR ALL $u\in V_j$ and all $r$-$u$ paths $P \in T$ DO 
     \begin{enumerate*}
     \item[(6)] FOR ALL $e \in \delta^+(u)$ DO
        \begin{enumerate*}
        \item[(7)] independently with prob.  $\frac{y_{P \cup \{ e\}}}{y_P}$ add $P \cup \{ e \}$ to $T$ 
        \end{enumerate*}
     \end{enumerate*}
  \end{enumerate*}
\item[(8)] return $E(T)$.
\end{enumerate*}
With $V(P)$ we denote the set of vertices on path $P$. 
Furthermore, let  $E(T) := \bigcup_{P \in T} P$ be the set of all edges on any path of $T$. 
Also let $V(T) := \bigcup_{P \in T} V(P)$.
Note that we did not remove partial paths (i.e. paths from $r$ to some layer $j < \ell$), which
will turn out to be convenient later.


\section{The analysis}

The analysis consists of two parts: 
\begin{enumerate} 
\item[(i)] We show that for each edge $e$ the probability to be included is $\Pr[e \in E(T)] \leq y_{\{ e\}}$.
\item[(ii)] We prove that for each terminal $s \in X$, the probability to be connected by a path satisfies $\Pr[s \in V(T)] \geq \Omega(\frac{1}{\ell})$. 
\end{enumerate}
Part $(i)$ provides that the expected cost for the sampled paths is at most $OPT_f$, while 
part $(ii)$ implies that after repeating the sampling procedure $O(\ell \log |X| )$ times, each terminal
will be connected to the root with high probability. 
Let us begin with part $(i)$. 

\subsection*{Upper bounding the expected cost}

For each node  $v \in V$, let $Q(v) := \{ P \mid P\textrm{ is }r\textrm{-}v\textrm{ path} \}$ be the set of
paths from the root to $v$. 
For an edge  $e=(u,v) \in E$ we denote $Q(e)$ as the set of $r$-$v$ paths that have $e$ as
last edge. 

\begin{lemma} \label{lem:Pr-P-in-T-is-yP}
Let $P$ be an $r$-$v$ path with $v\in V$. Then $\Pr[P \in T] = y_P$.
\end{lemma}
\begin{proof}
Let  $P = (e_1,\ldots,e_j)$ be the path with $e_i \in V_{i-1} × V_{i}$.
Then the probability that the algorithm samples path $P$ is
\[
  \Pr[P \in T] = y_{\{e_1\}}\cdot \frac{y_{\{e_1,e_2\}}}{y_{\{e_1\}}} \cdot \frac{y_{\{e_1,e_2,e_3\}}}{y_{\{e_1,e_2\}}} \cdot \ldots \cdot \frac{y_P}{y_{P \backslash \{e_j\}}} = y_P.
\]
\end{proof}

The next lemma will imply that each edge $e$ is sampled with probability
at most its fractional value $y_{\{e\}}$:
\begin{lemma} \label{lem:Sum-of-yP-PinQe-at-most-ye}
For any edge $e \in E$, one has $  \sum_{P \in Q(e)} y_P \leq y_{\{e\}}$.
\end{lemma}
\begin{proof}
We prove the following claim by induction over $j=0,\ldots,\ell-1$: \emph{For 
any edge $e \in V_j × V_{j+1}$ and any solution $\bar{Y} \in \textsc{Las}_{t'}(K)$
with $t' \geq j$ one has $ \sum_{P\in Q(e)} \bar{y}_P \leq \bar{y}_{\{e\}}$.}\footnote{We abbreviate $\bar{y}_P := \bar{Y}_{P,\emptyset}$.}

The claim is clear for $j=0$, thus consider an edge
 $e=(u,v)\in V_{j} × V_{j+1}$  between the $j$th and the $(j+1)$th level. 
Applying Thm.~\ref{thm:PropertiesOfLasT}.(\ref{item:PropertiesOfLasT-IntOnI}) 
with $I := \{ e \}$ we write 
$Y = y_{\{e\}}\cdot Y^{(1)} + (1-y_{\{e\}})\cdot Y^{(0)}$ such that $Y^{(0)},Y^{(1)} \in \textsc{Las}_{t'-1}(K)$ and $y_{\{e\}}^{(1)}=1$ as well as $y_{\{e\}}^{(0)} = 0$.
For edges  $e' \in \delta^-(u)$ ingoing to $u$, we apply the induction hypothesis and get $\sum_{P \in Q(e')} y_P^{(1)} \leq y_{\{e'\}}^{(1)}$. 
Since $y_{\{e\}}^{(1)}=1$, we know that $y_{\{e\} \cup P}^{(1)} = y_P^{(1)}$ (see Theorem~\ref{thm:PropertiesOfLasT}.\eqref{item:PropertiesOfLasT-yI1-implies-yIJ-is-yJ}). It follows that
\begin{eqnarray*}
\sum_{P \in Q(e)} y_P &=& y_{\{e\}} \cdot \sum_{P \in Q(e)} y_P^{(1)} = y_{\{e\}}\sum_{e' \in \delta^-(u)} \underbrace{\sum_{P \in Q(e')} \underbrace{y_{P \cup \{e\}}^{(1)}}_{=y_P^{(1)}}}_{\leq  y_{\{e'\}}^{(1)}}
\leq y_{\{e\}} \underbrace{\sum_{e \in \delta^-(u)} y_{\{e'\}}^{(1)}}_{\leq1} \leq y_{\{e\}}.
\end{eqnarray*}
\end{proof}
Combining both Lemmas~\ref{lem:Pr-P-in-T-is-yP} and \ref{lem:Sum-of-yP-PinQe-at-most-ye}, 
we obtain $\Pr[e \in E(T)] \leq y_{\{e\}}$ and consequently $E[c(E(T))] \leq \sum_{e\in E} c_e y_{\{e\}}$ by linearity of expectation.

\subsection*{Lower bounding the success probability}

In the following Lemma, we relate the ``path variables'' $y_P$ with the 
edge capacities $y_{\{e\}}$. In fact, $(a)$ will imply that each terminal is connected once in expectation and $(b)$ 
bounds the probability that $s$ is connected by a path containing a fixed subpath $P'$:
\begin{lemma}  \label{lem:SumOver-yPs}
Fix a terminal $s \in X$ and an $r$-$v$ path  $P'$ for some $v \in V$. Then
\begin{enumerate} \addtolength{\itemsep}{-0.4\baselineskip}
\item[a)] $\sum_{P \in Q(s)} y_P = 1$
\item[b)] $\sum_{P \in Q(s): P' \subseteq P} y_{P} \leq y_{P'}$.
\end{enumerate}
\end{lemma}
\begin{proof}
Consider the set of variables  $S := \{ f_{s,e} \mid e \in E\}$. If more than $\ell$ of those variables are set to $1$, 
they cannot define a feasible unit flow. In other words, we can apply Theorem~\ref{thm:PropertiesOfLasT}.(\ref{item:PropertiesOfLasT-DecThm}) in order to write $Y = \sum_{H \subseteq E} \lambda_HY^H$ as a convex combination of vectors 
 $Y^H$ such that for all $H$ with $\lambda_H>0$ one has: 
(i) $Y^H \in \textsc{Las}_{\ell}(K)$; (ii) $f_{s,e}^H \in \{ 0,1\}$ for all $e \in E$ and (iii) $f_{s,e}^H = 1 \Leftrightarrow e \in H$  (again we abbreviate  $f_{s,e}^H := Y_{\emptyset,(s,e)}^H$ and $y_{e}^H := Y_{\{e\},\emptyset}^H$).

But the variables $\{ f_{s,e}^H \mid e \in E\}$ can only represent a unit $r$-$s$ flow if $H$ is an $r$-$s$ path
as well. 
Thus we have $\lambda_H = 0$, whenever this is not the case. In other words, our convex combination is 
of the form $Y = \sum_{P \in Q(s)} \lambda_PY^P$.

Using Theorem~\ref{thm:PropertiesOfLasT}.(\ref{item:PropertiesOfLasT-Monotone}) we obtain $y_e^P \leq 1$
for all $e\in P$ and the LP constraints imply $y_e^P \geq f_{s,e}^P = 1$.
We conclude that $y_{e}^P = 1$ for all $e \in P$.
Then Theorem~\ref{thm:PropertiesOfLasT}.(\ref{item:PropertiesOfLasT-yI1-iff-yi1-for-all-i-in-I}) provides $y_P^P = 1$.

Conversely, consider any  $r$-$s$ path  $P' \in Q(s)$ with $P' \neq P$ and let $v \in V(P)$ be a vertex, where
path $P'$ enters $P$, i.e. $P \cap \delta^-(v) = \{ e\}$ and $P' \cap \delta^-(v) = \{ e'\}$
with $e\neq e'$. Since $y_e^P=1$ and $\sum_{e''\in\delta^-(v)} y_{e''}^P \leq 1$ (by LP constraint), 
we have  $y_{e'}^P = 0$ and thus $y_{P'}^P = 0$.
We conclude Claim $a)$, since
\[
y_P = \sum_{\bar{P} \in Q(s)}\lambda_{\bar{P}}  \underbrace{y_P^{\bar{P}}}_{=0\textrm{ if }P\neq\bar{P}} =\lambda_P
\]
and $\sum_{P \in Q(s)} \lambda_P=1$.

To see  $b)$ note that $y_{P'}^P = 1$, whenever $P'\subseteq P$. Thus
\[
y_{P'} = \sum_{P \in Q(s)} \lambda_P y_{P'}^P \geq \sum_{P \in Q(s): P' \subseteq P} \lambda_P = \sum_{P \in Q(s): P' \subseteq P} y_{P}.
\]
\end{proof}


%

In the following, we fix a terminal $s \in X$ and define $Z := |T \cap Q(s)|$ as the random variable 
that yields the number of sampled paths that end in $s$. 
Our goal is to show that $\Pr[Z \geq 1] \geq \Omega(\frac{1}{\ell})$. 
Recall that by Lemma~\ref{lem:SumOver-yPs}.(a) and Lemma~\ref{lem:Pr-P-in-T-is-yP}, we know already that
\begin{corollary}
$E[Z] = 1$.
\end{corollary}
Interestingly, the key insight of Garg, Konjevod and Ravi~\cite{GroupSteinerTreePolylogApx-GargKonjevodRavi-2000} 
is to prove an \emph{upper} bound on $Z$ in order to \emph{lower} bound $\Pr[Z \geq 1]$.
\begin{lemma} \label{lem:EX-at-most-l-plus-1}
$E[Z \mid Z \geq 1] \leq \ell+1$.
\end{lemma}
\begin{proof}
Fix a path $P = (e_1,\ldots,e_{\ell}) \in Q(s)$. It suffices to show $E[Z \mid P \in T] \leq \ell+1$.\footnote{
The formal argument works as follows: 
Let $A_1,\ldots,A_m$ be any events (in our application, $A_P$ is
the event ``\emph{$P \in T$, conditioned on $Z\geq1$}'')
and $Z := |\{ i \mid A_i\}|$ the number of occurring events. 
We claim that 
$E[Z] \leq \max_{i \in [m]} E[Z \mid A_i] =: \rho$. 
Proof: Using Jensen's inequality $
E[Z]^2\leq E[Z^2] = \sum_{i,j} \Pr[A_i \cap A_j] 
= \sum_i \Pr[A_i] \cdot \sum_{j} \Pr[A_j \mid A_i]
= \sum_i \Pr[A_i] \cdot E[Z \mid A_i]
\leq \rho \sum_i \Pr[A_i] = \rho E[Z]$. Rearranging yields the claim. 
}
Let  $P_i = (e_1,\ldots,e_i) \subseteq P$ be the $r$-subpath of $P$ containing the first $i$ edges.

Consider any path $P' \in Q(s)$ and say it contains $P_i$, but not $P_{i+1}$. 
Since the probability distribution depends only on the ``joint history'' of $P$ and $P'$, 
we have $\Pr[P' \in T \mid P \in T] = \Pr[P' \in T\mid P_i \in T]$.
We use this to bound
\begin{eqnarray*}
  E[|\{P' \in T\cap Q(s) \mid P_i\subseteq P'; P_{i+1} \nsubseteq P' \}|]
&\leq& \sum_{P' \in Q(s): P'\supseteq P_i} \Pr[P' \in T \mid P_i \in T] \\
&\stackrel{\textrm{cond. prob.}}{=}& \sum_{P' \in Q(s):P'\supseteq P_i} \frac{\Pr[P' \in T\textrm{ and }P_i \in T]}{\Pr[P_i \in T]} \\
&\stackrel{\textrm{Lemma~\ref{lem:Pr-P-in-T-is-yP}}}{=}& \sum_{P' \in Q(s):P'\supseteq P_i} \frac{y_{P'}}{y_{P_i}} \stackrel{\textrm{Lemma~\ref{lem:SumOver-yPs}}}{\leq} 1
\end{eqnarray*}
The claim follows since there are only $\ell+1$ such paths $P_i$.
\end{proof}
Garg-Konjevod-Ravi \cite{GroupSteinerTreePolylogApx-GargKonjevodRavi-2000} 
make use of a sophisticated probabilistic result, the \emph{Janson Inequality} (see e.g. \cite{ProbabilisticMethod-AlonSpencer08}).
However, the desired bound can be achieved much easier: 
\begin{lemma}
$\Pr[Z \geq 1] \geq \frac{1}{\ell+1}$.
\end{lemma}
\begin{proof}
By the law of total probability
\[
1 = E[Z] = \Pr[Z=0]\cdot\underbrace{E[Z \mid Z=0]}_{=0} + \Pr[Z \geq 1]\cdot \underbrace{E[Z \mid Z \geq 1]}_{\leq \ell+1\textrm{ by Lem.~\ref{lem:EX-at-most-l-plus-1}}}
\]
thus $\Pr[Z \geq 1] \geq \frac{1}{\ell+1}$.
\end{proof}



Finally, we show the $O(\ell\log |X| )$ integrality gap. Interestingly, though the 
Lasserre solution $Y$ has $n^{\Theta(\ell)}$ entries, we only query a polynomial number of entries $y_P$. 
In other words, if we could evaluate each single entry $y_P$ in polynomial time, the algorithm
would be polynomial as well.
\begin{theorem}
Let  $Y \in \textsc{Las}_{t}(K)$ be a given  $t=2\ell$ round Lasserre solution. 
Then one can compute a feasible solution $H \subseteq E$
with $E[c(H)] \leq O(\ell \log |X|) \cdot \sum_{e \in E} y_{\{e\}}$. The  expected number of Lasserre queries 
and the expected overhead running time 
are both polynomial in $n$.
\end{theorem}
\begin{proof}
Repeat the sampling algorithm for $2\ell\log |X|$ many times and let $H$ be the union of the sampled paths. 
The probability that a fixed terminal $s \in X$ is not connected is bounded by $(1 - \frac{1}{\ell+1})^{2\ell\log |X|} \leq \frac{1}{|X|}$.
For any terminal $s$ that remains unconnected, we buy the cheapest $r$-$s$ path. The expected cost for
this repair step is bounded by $|X| \cdot \frac{1}{|X|} \cdot \sum_{e \in E} c_{e}y_{\{e\}}$.

Consider a single sample $T$. 
For every node $v \in V$ (no matter whether terminal or not), the expected number of paths $P \in T$ 
connecting $v$ is upper bounded by one. Thus the expected number of sampled partial paths is 
bounded by $n$ (if we denote the total number of nodes on all layers by $n$). After a partial 
path $P$ is sampled, the algorithm queries at most $n$ values of the form $y_{P \cup \{e\}}$.
Thus the total expected number of queries for a single sample is upper bounded by $n^2$.
\end{proof}
Together with Lemma~\ref{thm:lLevelTree} and the fact that $Y$ can be computed in 
time $n^{O(\ell)}$, this provides a polynomial time $|X|^{\varepsilon}$-approximation algorithm
for any constant $\varepsilon > 0$. 
If we choose $\ell = \log |X|$, then we obtain a $O(\log^3 |X|)$ approximation in time $n^{O(\log |X|)}$.

\begin{remark}
Observe that we explicitly used the Decomposition Theorem of~\cite{IntGapLasserre-for-Knapsack-IPCO2011} in the proof of Lemma~\ref{lem:SumOver-yPs}. Since the Decomposition Theorem does not hold 
for the Sherali-Adams or Lovász-Schrijver hierarchy, it is not clear whether the same integrality gap
bound is true for those weaker relaxations. 

However, there is a well-known reduction from a level-$\ell$ instance of {\sc Directed Steiner Tree} 
to a tree instance $F$ of {\sc Group Steiner Tree} such that the produced tree $F$ has size $n^{O(\ell)}$ and 
contains all possible integral DST solutions as subtree. Of course, the corresponding 
{\sc Group Steiner Tree} LP for this instance $F$ has only a polylogarithmic 
integrality gap~\cite{GroupSteinerTreePolylogApx-GargKonjevodRavi-2000} and can also be
interpreted as an LP for {\sc Directed Steiner Tree}.
\end{remark}
%

It remains a challenging open problem, whether there is a convex relaxation with a 
$\textrm{polylog}(|X|)$ integrality gap that can be solved in polynomial time. 
Note that it would in fact suffice, to have a polynomial time oracle that takes a single
path $P\subseteq E$ as input and outputs the Lasserre entry $y_P$.

\paragraph{Acknowledgements.}

The author is very grateful to David Pritchard for carefully reading
a preliminary draft and to Michel X. Goemans, Neil Olver, Rico Zenklusen,
 Mohit Singh and David Steurer for helpful discussions and remarks. 

\bibliographystyle{alpha}
\bibliography{DirectedSteinerTreeViaLasserre}

\appendix

\section{Properties of the Lasserre Hierarchy}

The main goal of this section is to present a complete proof of the 
convergence of the Lasserre hierarchy (Theorem~\ref{thm:PropertiesOfLasT}.\eqref{item:PropertiesOfLasT-Hierarchy}); 
the feasibility of ``conditioned'' solutions (Theorem~\ref{thm:PropertiesOfLasT}.\eqref{item:PropertiesOfLasT-IntOnI})
and the Karlin-Mathieu-Nguyen Decomposition Theorem~\cite{IntGapLasserre-for-Knapsack-IPCO2011} (Theorem~\ref{thm:PropertiesOfLasT}.\eqref{item:PropertiesOfLasT-DecThm}). 
In the following, we are going to reproduce
the proof in \cite{IntGapLasserre-for-Knapsack-IPCO2011}.
However, we will use a different notation and
try to put the emphasis on an intuitive exposition instead of a space efficient one.

Let $\mathcal{P}([n]) := \mathcal{P}_n([n]) = \{ I \mid I \subseteq [n]\}$
be the family of all subsets of $[n]$. Recall that $K = \{ x \in \setR^n \mid Ax \geq b\}$ is the set of relaxed 
solutions with $A \in \setR^{m × n}$ and $b \in \setR^m$.

\subsection{The Inclusion-Exclusion Formula}

Suppose for the moment, that  $y \in \setR^{\mathcal{P}([n])}$ indeed is \emph{consistent}, 
i.e. there exists a random variable $Z \in \{ 0,1\}^n$
such that $\Pr[\bigwedge_{i\in I} (Z_i=1)] = y_I$ (thus  $E[Z_i]=y_i$). 
Here, $Z$ can be taken from any distribution --- especially 
 $Z_i$ and $Z_{i'}$ do not need to be independent. 

Initially, the Lasserre relaxation contains only variables for ``positive'' events
of the form $\bigwedge_{i\in I} (Z_i=1)$. But using the \emph{inclusion-exclusion formula}, one 
can also obtain probabilities for all other events. 

Recall that for any index set $J \subseteq [n]$ the inclusion-exclusion formula says that 
\[
  \Pr\Big[\bigvee_{i\in J} (Z_i = 1)\Big] = \sum_{\emptyset \subset H \subseteq J} (-1)^{|H|+1} \Pr\Big[\bigwedge_{i \in H} (Z_i=1)\Big].
\]
Negating this event yields 
\begin{equation} \label{eq:SimpleInclusionsExclusionsFormula}
  \Pr\Big[\bigwedge_{i \in J} (Z_i=0)\Big] = 1 - \Pr\Big[\bigvee_{i\in J} (Z_i = 1)\Big]
= \sum_{H \subseteq J} (-1)^{|H|} \Pr\Big[\bigwedge_{i \in H} (Z_i=1)\Big]. 
\end{equation}
Next, let $I \subseteq [n]$ be another index set (not necessarily disjoint to $J$). 
Observe that Equation~\eqref{eq:SimpleInclusionsExclusionsFormula} remains valid if
all events are intersected with the same event $\bigwedge_{i \in I} (Z_i=1)$. In other words
we arrive at the \emph{generalized inclusion exclusion formula} (sometimes called \emph{Möbius inversion})
\begin{equation} \label{eq:InklusionsExklusionsFormel}
  \Pr\Big[\bigwedge_{i \in I} (Z_i=1) \land \bigwedge_{i \in J} (Z_i=0)\Big] = \sum_{H \subseteq J} (-1)^{|H|} \Pr\Big[\bigwedge_{i \in I \cup H} (Z_i=1)\Big].
\end{equation}
Thus for any  $I,J \subseteq [n]$ we define
\begin{equation} \label{eq:DefY-I-minus-J}
 y_{I,-J} := \sum_{H \subseteq J} (-1)^{|H|} y_{I \cup H} = \Pr\Big[\bigwedge_{i \in I} (Z_i = 1), \bigwedge_{i \in J} (Z_i=0)\Big]
\end{equation}
(for example $y_{\emptyset,-\{1\}}= y_{\emptyset} - y_{\{1\}}$). 
If $I \dot{\cup} J = [n]$, then we abbreviate $y_{I,-J} =: y_x$ as the probability for the
\emph{atomic event} $x \in \{ 0,1\}^n$ with
\[
  x_i = \begin{cases} 1 & i \in I \\
  0 & i \in J
\end{cases}
\]
Furthermore we denote $\supp(x) := \{ i\in[n] \mid x_i = 1\}$ and $\bsupp(x) := \{ i\in[n] \mid x_i=0\}$. 
We saw so far, that the $2^n$ many probabilities  $y_I$ uniquely define the  $2^n$ many
 probabilities $y_x$ for atomic events. Conversely, one can obtain the values $y_I$ and $y_{I,-J}$ 
by summing over all atomic events that are consistent with the events, i.e.
\begin{eqnarray} 
  y_I &=& \sum_{x \in \{ 0,1\}^n: I \subseteq \supp(x)} y_x  \label{eq:yI-is-sum-of-yx} \\
 y_{I,-J} &=& \sum_{x \in \{ 0,1\}^n: I \subseteq \supp(x), J \subseteq \bsupp(x)} y_x. \label{eq:yIminusJ-is-sum-of-yx}
\end{eqnarray}
Let us make a couple of observations: 
\begin{itemize}
\item Equations~\eqref{eq:DefY-I-minus-J} and \eqref{eq:yI-is-sum-of-yx} are both linear, 
thus they define an isomorphism between $(y_I)_{I\subseteq[n]}$ and $(y_{x})_{x \in \{ 0,1\}^n}$.
This isomorphism is well defined even if $y$ is not consistent (i.e. even if
some $y_x$ are negative or $\sum_{x \in \{ 0,1\}^n} y_x \neq 1$).
\item If $I \cap J \neq \emptyset$, then by definition  one has 
$y_{I,-J} = 0$ since the sum in Eq.~\eqref{eq:DefY-I-minus-J} can be grouped into pairs
that have the same absolute value -- but different signs
(for example $y_{\{1\},-\{1,2\}} = y_{\{1\}\cup\emptyset} - y_{\{1\}\cup\{1\}} - y_{\{1\}\cup\{2\}} + y_{\{1\} \cup \{1,2\}}=0$).
\end{itemize}

\begin{remark}
Lemma~2 in the survey of Laurent~\cite{SDP-hierarchies-Survey-Laurent2003} states
the following equivalences for $y \in \setR^{\mathcal{P}([n])}$: 
\begin{itemize}
\item $M_n(y) \succeq 0 \; \Longleftrightarrow \; \forall x \in \{ 0,1\}^n: y_x \geq 0$
\item $M_n(\tchoose{A_{\ell}}{b_{\ell}} * y) \succeq 0 \; \Longleftrightarrow \; \forall x\in\{ 0,1\}^n: y_x \cdot (A_{\ell}x - b_{\ell}) \geq 0$
\end{itemize}
If both conditions hold and additionally $y_{\emptyset}=1$, then  $\sum_{x \in \{ 0,1\}^n} y_x = y_{\emptyset} = 1$ and the $y_x$ 
define a probability distribution over $\{ 0,1\}^n$ with $y_x=0$ for all $x$ with 
 $Ax \ngeq b$.
In other words
\[
 (y_{\{1\}},\ldots,y_{\{n\}}) = \sum_{x \in \{ 0,1\}^n: Ax \geq b} y_x \cdot x
\]
is a convex combination of feasible points.
However, we will show the convergence proof following \cite{IntGapLasserre-for-Knapsack-IPCO2011}, which has
more synergy effects with the decomposition theorem. 
\end{remark}

\subsection{Partial assignments and the inversion formula}

Let $\mathcal{T} \subseteq \mathcal{P}([n])$ be a family of index sets and let  $y \in \setR^{\mathcal{T}}$ be a corresponding vector. 
For subsets  $X \subseteq S \subseteq [n]$, we define the \emph{conditioning on $X$ and $-S\backslash X$} (sometimes also called \emph{partial assignment}) as the 
vector  $z = \{ y \}_{X, -S\backslash X} \in \setR^{\mathcal{T} \ominus S}$ with  $z_I := y_{I \cup X, -S\backslash X} = \sum_{H \subseteq S \backslash X} (-1)^{|H|}y_{I\cup X\cup H}$ and
$\mathcal{T} \ominus S := \{ I \in \mathcal{T} \mid \forall J\subseteq S: I \cup J \in \mathcal{T} \}$.
The definition of entry $z_I$ only makes sense, if $I \cup J \in \mathcal{T}$ for all $J \subseteq S$, 
thus $\mathcal{T} \ominus S$ is the maximum family of sets, for which this is satisfied.
Note that the set $\mathcal{T} \ominus S$ is not necessarily smaller
than $\mathcal{T}$. For example $\mathcal{P}([n]) \ominus S = \mathcal{P}([n])$ and $\mathcal{P}_t([n]) \ominus S \subseteq \mathcal{P}_{t - |S|}([n])$.

We define the \emph{normalized conditioning on $X$ and $-S\backslash X$} as $w:= \frac{z}{z_{\emptyset}}$ 
(say $w := \mathbf{0}$ if $z_{\emptyset} = 0$ to be well-defined). 
The intuition is that if $Z \in \{ 0,1\}^n$ again is a random variable with $y_{I,-J} = \Pr[\bigwedge_{i\in I} (Z_i=1), \bigwedge_{i\in J} (Z_i=0)]$,
then the (normalized) conditioned solution reflects conditional probabilities, i.e.
\[
w_{I,-J} = \Pr\Big[\bigwedge_{i\in I} (Z_i=1), \bigwedge_{i\in J} (Z_i=0) \mid \bigwedge_{i\in X} (Z_i=1), \bigwedge_{i\in S\backslash X} (Z_i=0)\Big].
\]

The events $(X, - S\backslash X)$ obviously partition the probability space, if $X$ runs over all subsets of $S$. 
This remains valid for conditioned Lasserre solutions. 

\begin{lemma}[Inversion formula]  \label{lem:InversionFormula}
Let $y \in \setR^{\mathcal{P}([n])}$ and  $S \subseteq [n]$. Then  $y = \sum_{X \subseteq S} \{ y \}_{X,-S\backslash X}$.
\end{lemma}
\begin{proof}
We verify the equation for entry $I \subseteq V$: 
\begin{eqnarray*}
\sum_{X \subseteq S} y_{I \cup X, -S \backslash X} &=& \sum_{X \subseteq S} \sum_{x \in \{ 0,1\}^n: I \cup X \subseteq \supp(x), S \backslash X \subseteq \bsupp(x)} y_x \\
&=& \sum_{x \in \{ 0,1\}^n: I \subseteq \supp(x)} y_x \cdot \underbrace{|\{ X \subseteq S: X \subseteq \supp(x), S \backslash X \subseteq \bsupp(x) \}|}_{= 1} = y_I
\end{eqnarray*}
\end{proof}

\subsection{Feasibility of conditioned solutions}

Next, we will see that conditioned solutions are still feasible on a smaller family of index sets: 
\begin{lemma} \label{lem:MYPsd-impliziert-MZpsd}
Let $X \subseteq S \subseteq V$, 
and $y \in \setR^{\mathcal{P}([n])}$ with  $\mathcal{T} \subseteq \mathcal{P}(V)$. Then
\[
  M_{\mathcal{T}}(y) \succeq 0 \quad \Longrightarrow \quad M_{\mathcal{T}\ominus S}(\{y\}_{X,-S\backslash X}) \succeq 0
\]
\end{lemma}
\begin{proof}
Abbreviate $z := \{ y_{X,-S\backslash X}\} \in \setR^{\mathcal{P}([n])}$. 
Equivalently to $M_{\mathcal{T}}(y) \succeq 0$, there must be vectors $v_I$  with $v_I\cdot v_J = y_{I \cup J}$ for $I,J \in \mathcal{T}$.
For each $I \in \mathcal{T} \ominus S$, we define another vector
\[
 w_I := \sum_{H \subseteq S \backslash X} (-1)^{|H|} v_{I \cup X \cup H}.
\]
We claim that those vectors provide a factorization of  $M_{\mathcal{T} \ominus S}(z)$ (which proves the positive semi-definiteness of $M_{\mathcal{T} \ominus S}(z)$)
\begin{eqnarray*}
w_I\cdot w_J &=& \sum_{H \subseteq S \backslash X} \sum_{L \subseteq S \backslash X} (-1)^{|H|+|L|} v_{I \cup X \cup H}\cdot v_{J \cup X \cup L} \\
&=& \sum_{H \subseteq S \backslash X} \underbrace{\sum_{L \subseteq S \backslash X} (-1)^{|H|+|L|} y_{I \cup J \cup X \cup H \cup L}}_{(*) = 0\textrm{ if }H \neq \emptyset} \\
&=& \sum_{L \subseteq S \backslash X}  (-1)^{|L|} y_{I \cup J \cup X \cup L} =  y_{I \cup J \cup X, -S \backslash X} = z_{I \cup J}.
\end{eqnarray*}
To see $(*)$, observe that if there is any $i \in H \subseteq S\backslash X$, then the sum contains 
the same term for 
 $L \subseteq (S\backslash X)\backslash\{i\}$ and $L\cup\{i\}$ -- just with different sign, so that the sum evaluates to $0$.  
\end{proof}

The following lemma is usually called \emph{commutativity of the shift operator} (see e.g. \cite{SDP-hierarchies-Survey-Laurent2003}). 
Recall that for  $y \in \setR^{\mathcal{T}}$, we interpret $w = \tchoose{a}{\beta} * y$ as the vector %
with $w_I := (\sum_{i\in[n]} a_iy_{I\cup\{i\}}-\beta y_I) $.\footnote{In fact, we were sloppy concerning the dimension of $w$ so far. Formally, one should define $w \in \setR^{\mathcal{T}'}$ with $\mathcal{T}' := \{ I \mid I \cup \{i\} \in \mathcal{T} \; \forall i \in [n]\}$. }  
\begin{lemma} \label{lem:CommutativityOfShiftOp}
Let $y \in \setR^{\mathcal{T}}$; $X \subseteq S \subseteq [n]$; 
and $a^Tx \geq \beta$ be a linear constraint. Then
\[
  \tchoose{a}{\beta} * \underbrace{\{y\}_{X,-S\backslash X}}_{=:z} = \{\underbrace{\tchoose{a}{\beta} * y}_{=:u}\}_{X,-S\backslash X}   \quad (\in \setR^{\mathcal{T}'})
\]
with $\mathcal{T}' := \{ I \subseteq [n] \mid I \cup J \cup \{i\} \in \mathcal{T} \; \forall J \subseteq S \; \forall i \in [n]\}$.
\end{lemma}
\begin{proof}
Let $z_I = y_{I \cup X,-S\backslash X}$ and $u_I = \sum_{i \in [n]} a_iy_{I\cup\{i\}}-y_I\beta$. 
Evaluating the left hand side vector at entry $I\in \mathcal{T}'$ gives
\begin{eqnarray*}
  (\tchoose{a}{\beta} * \{y\}_{X,-S\backslash X})_I &=& \sum_{i\in[n]} a_iz_{I \cup\{i\}} -\beta z_I
= \sum_{i\in[n]} a_iy_{I \cup X \cup \{ i\}, -S \backslash X}-\beta y_{I\cup X,-S\backslash X} \\
&=& \sum_{H \subseteq S \backslash X} (-1)^{|H|} \Big[\sum_{i\in [n]} a_i y_{I \cup X \cup \{ i\} \cup H} - \beta y_{I\cup X\cup H}\Big]
\end{eqnarray*}
The right hand side entry for $I$ is 
\begin{eqnarray*}
  u_{I \cup X, -S\backslash X} &=& \sum_{H \subseteq S\backslash X} (-1)^{|H|} u_{X \cup I \cup H} \\
&=&  \sum_{H \subseteq S\backslash X} (-1)^{|H|} \Big[\sum_{i\in[n]} a_iy_{X \cup I \cup \{i\} \cup H}- \beta y_{I\cup X\cup H}\Big]
\end{eqnarray*}
Both expressions are identical and the claim follows.
\end{proof}

Now we will see that normalized conditioned solutions are feasible and integral on variables in $S$: 
\begin{lemma}  \label{lem:InducedSolutionFeasible}
Let $X \subseteq S \subseteq [n]$, $\mathcal{T}_1,\mathcal{T}_2 \subseteq \mathcal{P}([n])$ 
and $y \in \setR^{\mathcal{P}([n])}$ with
\[
  M_{\mathcal{T}_1}(y) \succeq 0; \quad M_{\mathcal{T}_2}(\tchoose{A_{\ell}}{b_{\ell}} * y) \succeq 0 \; \forall \ell\in[m]; \quad y_{\emptyset} = 1.
\]
Define $z := \{ y\}_{X,-S\backslash X} \in \setR^{\mathcal{P}([n])}$. If $z_{\emptyset} > 0$, then for $w := \frac{z}{z_{\emptyset}}$ one has
\begin{itemize*}
\item $M_{\mathcal{T}_1 \ominus S}(w) \succeq 0; \quad M_{\mathcal{T}_2 \ominus S}(\tchoose{A_{\ell}}{b_{\ell}} * w) \succeq 0 \; \forall \ell\in[m]; \quad w_{\emptyset} = 1.$
\item $w_I = 1$ if $I \subseteq X$
\item $w_I = 0$ if $I \subseteq S$ but $I \cap (S\backslash X) \neq \emptyset$.
\end{itemize*}
\end{lemma}
\begin{proof}
Lemma~\ref{lem:MYPsd-impliziert-MZpsd} directly implies that $M_{\mathcal{T}_1 \ominus S}(w) = \frac{1}{z_{\emptyset}} M_{\mathcal{T}_1 \ominus S}(z) \succeq 0$. Applying the commutativity rule,
\begin{eqnarray*}
  M_{\mathcal{T}_2 \ominus S}(\tchoose{A_{\ell}}{b_{\ell}} * w) &=& \frac{1}{z_{\emptyset}} M_{\mathcal{T}_2 \ominus S}(\tchoose{A_{\ell}}{b_{\ell}} * \{ y\}_{X,-S\backslash X}) \\
&\stackrel{\textrm{Lemma~\ref{lem:CommutativityOfShiftOp}}}{=}& \frac{1}{z_{\emptyset}} M_{\mathcal{T}_2 \ominus S}(\{ \tchoose{A_{\ell}}{b_{\ell}} *  y\}_{X,-S\backslash X}) \stackrel{\textrm{Lemma~\ref{lem:MYPsd-impliziert-MZpsd}}}{\succeq}  0
\end{eqnarray*}
using that $M_{\mathcal{T}_2}( \tchoose{A_{\ell}}{b_{\ell}} *  y ) \succeq 0$. Furthermore for $I \subseteq S$ one has
\[
w_{I} = \frac{z_{I}}{z_{\emptyset}} = \frac{y_{X\cup I,-S\backslash X}}{y_{X, -S\backslash X}} = \begin{cases} 1 & I \subseteq X \\ 0 & I \cap (S \backslash X) \neq \emptyset \end{cases}.
\]
\end{proof}

\subsection{Convergence}

With the last lemma at hand, the convergence of the Lasserre hierarchy follows quickly.
\begin{lemma}
$K \supseteq \textsc{Las}_0^{\textrm{proj}}(K) \supseteq \ldots \supseteq \textsc{Las}_n^{\textrm{proj}}(K) \supseteq \conv(K \cap \{ 0,1\}^n)$.
\end{lemma}
\begin{proof}
By definition $\textsc{Las}^{\textrm{proj}}_{t}(K) \supseteq \textsc{Las}^{\textrm{proj}}_{t+1}(K)$. 
Let $y \in {\textsc{Las}}_0(K)$. Then $0 \leq (M_0(\tchoose{A_{\ell}}{b_{\ell}} * y))_{\emptyset,\emptyset} = \sum_{i\in[n]} A_{\ell i}y_{\{i\}} - b_i$,
thus $A(y_{\{1\}},\ldots,y_{\{n\}}) \geq b$ and $K \supseteq \textsc{Las}_0^{\textrm{proj}}(K)$.

Finally let $x \in K \cap \{ 0,1\}^n$ and define $y \in \setR^{\mathcal{P}([n])}$ with 
$y_I := \prod_{i \in I} x_i$. Then one has $M_n(y) = yy^T \succeq 0$. Furthermore $M_n(\tchoose{A_{\ell}}{b_{\ell}} * y) = (A_{\ell}x-b_{\ell}) \cdot yy^T \succeq 0$, 
thus $y \in \textsc{Las}_n(K)$. 
\end{proof}

The following statement implies that $\textsc{Las}_n^{\proj}(K) = \conv(K \cap \{ 0,1\}^n)$:
\begin{lemma} \label{lem:Convergence}
$\textsc{Las}_n(K) \subseteq \conv\{ (\prod_{i\in I} x_i)_{I \subseteq [n]} \mid x \in \{ 0,1\}^n : Ax \geq b \}$.
\end{lemma}
\begin{proof}
For all $X\subseteq[n]$, define $z^{X} := \{ y \}_{X,-[n]\backslash X} \in \setR^{\mathcal{P}([n])}$ and $w^X := \frac{z^X}{z^X_{\emptyset}}$ if $z_{\emptyset}^X>0$. Then
\[
y \stackrel{\textrm{Lemma~\ref{lem:InversionFormula}}}{=} \sum_{X \subseteq [n]} z^X = \sum_{X \subseteq [n]: z^X_{\emptyset} > 0} z_{\emptyset}^X \cdot w^X. 
\]
This is a convex combination since $\sum_{X \subseteq [n]} z_{\emptyset}^X = y_{\emptyset} = 1$ (again by Lemma~\ref{lem:InversionFormula}). For a fixed $X$, we abbreviate $x_i := w^X_{\{i\}}$. 
Then Lemma~\ref{lem:InducedSolutionFeasible} provides that $w^X \in \textsc{Las}_n(K)$ (thus $x \in K$); $x \in \{ 0,1\}^n$ and $w^X_I = \prod_{i\in I} x_i$.
\end{proof}

\subsection{Local consistency}

Let $y \in \setR^{\mathcal{P}_{t}(V)}$ and  $t'<t$. Then $y_{|\mathcal{P}_{t'}(V)} \in \setR^{\mathcal{P}_{t'}(V)}$
is the vector that emerges from $y$ after deletion of all entries $I$ with $|I| > t'$. 
Moreover, for any vector $y \in \setR^{\mathcal{T}}$ with $\mathcal{T} \subseteq \mathcal{P}([n])$, we define
the \emph{extension} $y' \in \setR^{\mathcal{P}([n])}$ as the vector $y$ where missing entries 
are filled with zeros. 
\begin{lemma}
Let $y \in \textsc{Las}_{t}(K)$ and $S \subseteq [n]$ with $|S| \leq t$. Then 
\[
y \in \conv\{ w \mid w_{|\mathcal{P}_{2(t - |S| ) + 2}} \in {\textsc{Las}}_{t-|S|}(K); w_{\{i\}} \in \{ 0,1\} \; \forall i \in S \}
\]
\end{lemma}
\begin{proof}
Again we can write a convex combination $y' = \sum_{X \subseteq S} z_{\emptyset}^X \cdot w^X$ (with $w^X \in \setR^{\mathcal{P}(n)}$) according 
to Lemma~\ref{lem:InversionFormula}.  Recall that $\mathcal{P}_{2t+2}([n]) \ominus S \supseteq \mathcal{P}_{2(t-|S| ) + 2}([n])$,
thus Lemma~\ref{lem:InducedSolutionFeasible} provides that $w_{|\mathcal{P}_{2(t - |S| ) + 2}([n])} \in \textsc{Las}_{t - |S|}(K)$
and that $w$ is integral on $S$.
\end{proof}


\subsection{The decomposition theorem}

Imagine for a second that $y \in \textsc{Las}_t(K)$ and $y_I = 0$ for all 
$|I| \geq t$. Then we can just fill the matrices $M_{t+1}(y)$ and $M_{t}(\tchoose{A_{\ell}}{b_{\ell}} * y)$ with zeros to obtain $M_n(y)$ and $M_{n}(\tchoose{A_{\ell}}{b_{\ell}} * y)$ without destroying positive semi-definiteness or their consistency. 
Consequently, $y$ would even be in the convex hull of feasible integral vectors, 
even though we only assumed $y$ to be a $t$-round solution. 
With a bit more care, this approach applies more generally to any subset of variables. 
\begin{theorem}[Decomposition Theorem~\cite{IntGapLasserre-for-Knapsack-IPCO2011}]
Let $0\leq k\leq t$, $y \in {\textsc{Las}}_t(K)$, $X\subseteq S \subseteq V$ so that $|I \cap S| > k \Rightarrow y_I = 0$.
Then
\[
  y \in \conv\{ w \mid w_{|2(t-k)+2} \in \textsc{Las}_{t-k}(K); \; w_{\{i\}} \in \{ 0,1\} \; \forall i\in S\}
\]
\end{theorem}
\begin{proof}
Again extend $y$ to $y' \in \setR^{\mathcal{P}(n)}$. 
Define $\mathcal{T}_1 := \{ A \subseteq [n] \mid |A\backslash S| \leq t+1-k\}$ as the set of 
indices that have at most $t+1-k$ indices outside of $S$. 
After sorting the rows and columns by increasing size of index sets, we can write
\[
   \begin{array}{cc}
\small{\begin{matrix} \hspace{1.2cm} J: |J|\leq t+1 & J: |J| >t+1 \end{matrix}} & \\
M_{\mathcal{T}_1}(y') = \begin{pmatrix} \psframebox[framesep=3pt]{\hspace{0.5cm} M\hspace{0.5cm}} & \psframebox[framesep=5pt]{\hspace{0.5cm}\mathbf{0}\hspace{0.5cm}} \\
\psframebox[framesep=5pt]{\hspace{0.5cm}\mathbf{0}\hspace{0.5cm}} & \psframebox[framesep=5pt]{\hspace{0.5cm}\mathbf{0}\hspace{0.5cm}} 
\end{pmatrix}
\succeq 0 &
\end{array}
\]
where $M$ is a principal submatrix of $M_{t+1}(y) \succeq 0$.
Next, observe that there are entries $y_J'$ that may appear inside of $M$ and outside. 
If they appear outside, say
at entry $(J_1,J_2)$, then $J = J_1 \cup J_2$ and $|J_1| > t+1$, but $|J_1\backslash S| \leq t+1-k$ (so that $J_1 \in \mathcal{T}_1$).
Thus $|J \cap S| \geq |J_1 \cap S| = |J_1| - |J_1\backslash S| > k$, thus $y_J' = 0$ by assumption\footnote{In other words, the matrix $M_{\mathcal{T}_1}(y')$ is \emph{consistent} 
in the sense that entries $(J_1,J_2)$ and $(J_3,J_4)$ are identical whenever $J_1 \cup J_2 = J_3 \cup J_4$.}.

Analogously for $\mathcal{T}_2 := \{ A \subseteq [n] \mid |A\backslash S| \leq t-k\}$ one can write
\[
   \begin{array}{cc}
\hspace{1.5cm} \small{\begin{matrix} \quad J: |J| \leq t & \quad\;\; J: |J| > t \end{matrix}} & \\
M_{\mathcal{T}_2}(\tchoose{A_{\ell}}{b_{\ell}} * y') =\begin{pmatrix}  \psframebox[framesep=3pt]{\hspace{0.5cm}N\hspace{0.5cm}} & \psframebox[framesep=5pt]{\hspace{0.5cm}\mathbf{0}\hspace{0.5cm}} \\
\psframebox[framesep=5pt]{\hspace{0.5cm}\mathbf{0}\hspace{0.5cm}} & \psframebox[framesep=5pt]{\hspace{0.5cm}\mathbf{0}\hspace{0.5cm}} 
\end{pmatrix} \succeq 0
&
\end{array}
\]
with $N\succeq0$ being a principal submatrix of $M_{t}({A_{\ell} \choose b_{\ell}} * y) \succeq 0$.
Again, consider an entry $J$ that appears at least once outside of $N$ thus
we can write $J= J_1\cup J_2$ with $|J_1| > t$ and $J_1 \in \mathcal{T}_2$,
hence $|J_1 \backslash S| \leq t-k$. Then for any $J'\supseteq J_1$ one has $|J'\cap S| \geq |J_1 \cap S| = |J_1|-|J_1\backslash S| \geq (t+1)-(t-k)> k$ and $y_{J'}' = 0$ by assumption.
The matrix entry at position $J$ is $\sum_{i\in[n]} A_{\ell i}y_{J\cup\{i\}}-\beta y_{J} = 0$, 
hence the matrix is defined consistently.

Define $z^X := \{ y \}_{X,-S\backslash X} \in \setR^{\mathcal{P}([n])}$ and $w^X := \frac{z^X}{z_{\emptyset}^X}$ ($w^X:=\mathbf{0}$, if $z_{\emptyset}^X=0$), 
then as in Lemma~\ref{lem:Convergence}, we can express $y'$ as  convex combination $y' = \sum_{X \subseteq S} z_{\emptyset}^X w^X$. 
Since $\mathcal{T}_1 \ominus S = \mathcal{T}_1$ and $\mathcal{T}_2 \ominus S = \mathcal{T}_2$, we apply Lemma~\ref{lem:InducedSolutionFeasible} and obtain
 $M_{\mathcal{T}_1}(w^X) \succeq 0$,  $M_{\mathcal{T}_2}(\tchoose{A_{\ell}}{b_{\ell}} * w^X) \succeq 0$ for all $\ell \in [m]$ and $w^X_{\emptyset} = 1$.
Observe that  $\mathcal{P}_{t+1-k}([n]) \subseteq \mathcal{T}_1$ and $\mathcal{P}_{t-k}([n]) \subseteq \mathcal{T}_2$. Consequently $w_{|2(t-k)+2} \in \textsc{Las}_{t-k}(K)$ as claimed. 
\end{proof}
Now we argue, why the Decomposition Theorem implies Property~\ref{thm:PropertiesOfLasT}.(\ref{item:PropertiesOfLasT-DecThm}).
\begin{lemma*}[Property~\ref{thm:PropertiesOfLasT}.(\ref{item:PropertiesOfLasT-DecThm})]
Let $K = \{ x \in \setR^n \mid Ax \geq b\}$ and $y \in {\textsc{Las}}_t(K)$ and assume that
for a subset of variables $S \subseteq [n]$ one has $\max\{ |I| : I \subseteq S; x \in K; x_i = 1\; \forall i \in I\} \leq k < t$.
Then $y \in \conv(\{z \in \textsc{Las}_{t-k}(K) \mid z_{\{i\}} \in \{ 0,1\} \; \forall i\in S\})$.
\end{lemma*}
\begin{proof}
Consider an index set $I \subseteq S$ with $|I| = k+1\leq t$. Then Property~\ref{thm:PropertiesOfLasT}.(\ref{item:PropertiesOfLasT-Cond-for-yI0}) implies that $y_I = 0$. For all $|I|,|J| \leq t+1$ with $y_I=0$, inspecting
the determinant of the principal submatrix of $M_{t+1}(y)$ induced by indices $I$ and $J$, we see that
\[
  0\leq \det \begin{pmatrix} 0 & y_{I \cup J} \\ 
  y_{I \cup J} & y_J \end{pmatrix} = -y_{I \cup J}^2,
\]
thus $y_{I \cup J} = 0$ (i.e., all entries are monotone). 
We summarize: All entries $I \subseteq [n]$ with $|I| \leq 2(t+1)$
and $|I \cap S| > k$ have  $y_I=0$. Then by the Decomposition Theorem we have
\[
y \in \conv\{ w \mid w_{|2(t-k)+2} \in \textsc{Las}_{t-k}(K); \; w_{\{i\}} \in \{ 0,1\} \; \forall i\in S\}.
\]
\end{proof}
\end{document}